\documentclass[conference]{IEEEtran}[10pt]
\usepackage{amsmath}
\usepackage{amssymb}
\usepackage{graphicx}
\usepackage{subfigure}
\usepackage{a4wide}
\usepackage{algorithm}
\usepackage{algorithmic}
\usepackage{amsfonts}

\makeatletter
\def\ps@headings{%
\def\@oddhead{\mbox{}\scriptsize\rightmark \hfil \thepage}%
\def\@evenhead{\scriptsize\thepage \hfil \leftmark\mbox{}}%
\def\@oddfoot{}%
\def\@evenfoot{}}
\makeatother
\newtheorem{theorem}{Theorem}

\begin{document}
\title{A Security Architecture for Data Aggregation  and Access Control in Smart Grids}

\author{
\IEEEauthorblockN{Sushmita Ruj, Amiya Nayak and Ivan Stojmenovic\\}
\IEEEauthorblockA{
SEECS, University  of Ottawa,\\
Ottawa K1N 6N5, Canada   \\
Email: \{sruj, anayak, ivan\}@site.uottawa.ca}}

\maketitle{}

\begin{abstract}
We propose an integrated architecture for smart grids, that supports data aggregation  and access control.
Data can be aggregated by home area network, building area network and neighboring area network in such a way that the privacy of customers is protected.
We use homomorphic encryption technique to achieve this.
The consumer data that is collected is sent to the substations where it is monitored by remote terminal units (RTU).
The proposed access control mechanism
gives selective access to  consumer data
stored in data repositories and used by different smart grid users.
Users can be
maintenance units, utility centers, pricing estimator units or analyzing and prediction groups. 
We solve this problem of access control using cryptographic technique of attribute-based encryption.
RTUs  and  users have attributes and cryptographic keys distributed by several key distribution centers (KDC).
RTUs send data encrypted under a set of attributes.
Users can decrypt information provided they have valid attributes.
The access control scheme is distributed in nature and does not rely on a single KDC to distribute keys.
Bobba \emph{et al.} \cite{BKAA09} proposed an access control scheme, which relies on a centralized KDC and is thus prone to single-point failure.
The other requirement is that the KDC has to be online, during data transfer which is not required in our scheme.
Our access control scheme is collusion resistant, meaning that users cannot collude and gain access to data, when they are not authorized to access.
We theoretically analyze our schemes (with mathematical proofs of correctness) and
show that the computation overheads are low enough to be carried out in smart grids.
To the best of our knowledge, ours is the first work on smart grids, which integrates these two important security components (privacy preserving
data aggregation and access control) and presents an overall security architecture in smart grids.

\end{abstract}


\section{Introduction}
\label{sec:intro}
Smart grids are next generation electricity grid system which will integrate power and communication networks. 
With the growing demand for electricity, there is a need to develop smart grids which can cope up with the demand
by intelligently using different power resources and integrating different components like vehicles, and wireless devices. 
Smart grids should have capabilities that would enable it to deal with power outages by balancing supply and demand.
This can be achieved by intelligently balancing the consumption between peak and off-peak periods. 
One recent suggestion has been to charge electric vehicles (also incorporated into the grid) during the off-peak period and
discharge it back into the grid. In this way the grid is bi-directional, energy can be used when needed and discharged back
into the grid when not needed. 

The operation of smart grid involves many aspects:
generation of power using different sources like solar, wind, geothermal, nuclear, fossil-fuel, 
the intelligent distribution of power by monitoring the demand of power in different regions and different customers,
monitoring the power usage by customers using smart meters and intelligently deliver power when needed, 
building and integrating appliances into the grid, like vehicles (plugged in electric vehicles - PHEV) and wireless devices. 

Research in smart grid is very important and involves a broad range of problems. 
An important problem is to design an architecture integrating all the components which can efficiently use electricity.
Smart grid architectures have been proposed and discussed by Bose \cite{B10}.  
It comprises of power infrastructure and information infrastructure \cite{KTKL10}. 
Power infrastructure consists of power equipments like generators, transformers, transmission lines, voltage regulators, capacity banks, meters etc, 
which help to  deliver electricity. 
The power infrastructure involves generation of power from different sources and their reliable and efficient transmission. 
Energy efficient distribution of power is presented in \cite{WYJ09} and \cite{MMD10}. 

The information infrastructure helps in communication and ensures safety and reliability. 
It measures the status of the devices in the grid, balances demand and supply, helps in diagnosis of faults, helps authentication of devices
and helps in the smooth working of plugged in devices like vehicles. 
Devices might have sensors to sense different conditions and can be simple devices as smoke detectors and automatic light switches etc.
There are also devices called phasor measurement units (PMUs) which measure electrical waves in the grid. 
PMUs are clock synchronized (through GPS) sensors that can read current and voltage phasors at a substation bus on the transmission power network 
\cite{BKAA09}. These phasors can send 50-60 measurements per second \cite{GSA09}. 
Load balancing is an important aspect of research.  
Direct load control (DLC) \cite{RCO09} can remotely control appliances in homes and workplaces and reduce energy consumption. 
Game theoretic techniques are being increasingly used to optimize consumption. 
One way to do this is  consumption scheduling \cite{MWJSL10}.  

There is a huge economic aspect of smart grids and demands a lot of attention. 
This relates to pricing and marketing policies, legal and ethical issues. 
At one hand it is important to switch towards green energy like solar, wind etc and on the other hand it is important how to 
make best use of these renewable sources of energy and integrate them into the grid. 
It might be easier to use the energy close to the source to reduce transmission loss and costs. 
Several pricing policies are also being regulated by the government. These also require manual and ethical considerations. 
For example, reducing the consumption of electricity at peak hours. 
Critical-peak pricing (CPP), real-time pricing (RTP), time-of-use pricing (ToUP) are popular ways of reducing consumption. 
These policies impose different rates during different time  of the day (more during peak hours) or year (cold days in winter and hot days in summer). 

Control decisions of embedded systems in critical infrastructure can have significant impact on human life and the environment. 
Cyber physical systems need to combine computational decision making on the cyber side with physical control on the device side. The network that connects intelligent devices must ensure that critical data are available for making informed decisions. 
Smart grid with all its advantages must be fault tolerant,  reliable and secure. 
It should be possible to detect fault early in the system, to protect against cascading effects. 

Conventional power grids utilize centralized command and control structures, such as SCADA (Supervisory Control And Data Acquisition) systems relying on human monitors for decision making. 
SCADA systems provide the mechanism for identifying faults. However, they represent a single point of failure within today’s power grid. 
Further, even when SCADA systems are running with specified parameters, catastrophic faults (e.g. cascading failures) can occur \cite{ZM10}.
Detecting faults earlier in the network is extremely important because faults can easily propagate throughout the network and
lead to complete breakdown. Such a blackout occurred in August 2003, which affected 45 million people in US and 10 million people in Canada. 
The damages due to this blackout has been estimated as 6 billion US dollars. 
Thus, designing fault tolerant grid is very important. 
In this direction, it should be possible to divert the power to alternate route once a particular route is disrupted. 
Zimmer and Mueller \cite{ZM10} proposed a fault tolerant network routing through software overlays. 

An important problem which is associated with smart grid is the problem of security and privacy. 
It is very important to secure the smart grid, not only from terrorist attacks,
but also from customers, and building authorities who tamper with various devices.  
The information from remote terminal units (RTU) at the substation  is needed not only for electricity distribution, 
but also for calculating costs, for predicting future 
conditions and for monitoring in case of unexpected behavior. 
All these tasks are done by separate users, for example the electrical and maintenance board will monitor the network, 
the costs calculation and analysis is done by the auditing unit and to predict future behavior researchers can be involved. 
All information must be sent only to the users responsible for specific job. 
Access control thus becomes a very important issue in smart grids. 
In future, when content distribution will also be included into the smart grid (our assumption is that future smart grids will also have cable 
integrated into it),
it will be necessary to regulate the access, such that two or more users do not collude and access information they cannot individually access. 
Existing literature focus on either authentication authentication \cite{FFKLS10,CSSL11,KBYAH10} or
privacy protection \cite{RSMP11,SKTP11}. 
Surveys on security and related aspects of smart grids appear in \cite{MR10}.  

We present a security architecture that integrate privacy preserving data aggregation and access control for the first time. 
Data aggregation has been studied by Li \emph{et al.} \cite{LLL10}, however it is very limited in scope. 
It presents privacy protected data aggregation in a local neighborhood (typically a building area network) without focusing 
on large scale aggregation. 
It also does not say how keys are distributed and more concerned with efficient construction of data aggregation trees. 
Access control has been studied by Bobba \emph{et al.} \cite{BKAA09}. 
They proposed a policy based encryption scheme for access control in smart grids.
The main assumption is the existence of a fully honest key distribution center (KDC) who distributes keys and access policies to 
data senders and receivers. 
A receiver can decrypt information, if it has a valid set of attributes. 
The policies are implemented in XML and the encryption mechanism uses KEM-DEM hybrid encryption paradigm introduced by Cramer and Shoup \cite{CS04}. 
KDC distributes keys and access policies.

The scheme in \cite{BKAA09} is prone to failure if the single KDC is compromised. 
It also demands the KDC be online during data access, thus halting all activities during failure or maintenance. 
For reasons of efficiency and security, multiple KDCs is desirable. 
For this reason, we use multiple KDCs. Our access control scheme is based on
attribute based encryption protocol, which is being increasingly used for access control in different domains like clouds \cite{RNS11c}, 
ad hoc networks \cite{RNS11} etc. 

\begin{figure*}[htb]
\begin{centering}
\includegraphics[width=6in]{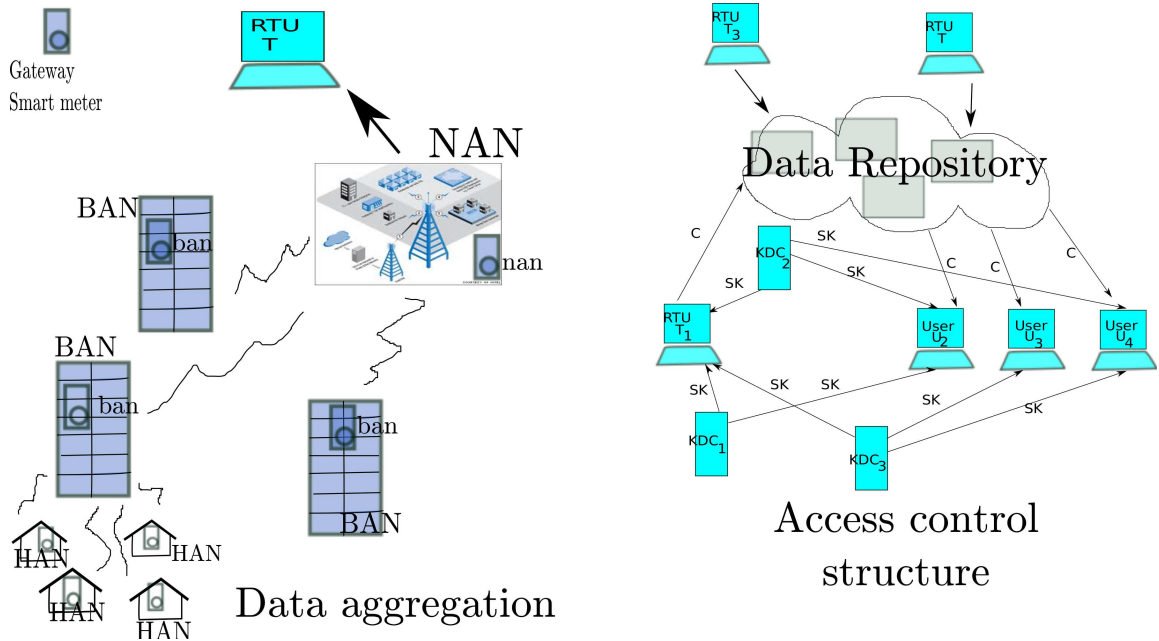}
\caption{
Aggregation and access control architecture
}
\label{fig:architecture}
\end{centering}
\end{figure*}

Our architecture consists of two parts, the first network consists of home area networks (HAN), building area network (BAN) and neighborhood
area network (NAN) which reports to a substation. 
For each home area network there is a gateway smart meter $han$ which collects information and sends to  
the building area network. The gateway $ban$ aggregates all information from smart meters in the BAN and sends to the
$nan$ at the neighborhood area. 
$nan$ reports to the substation. 
		
The second part consists of the RTU at the substation
who send aggregated results to data centers for storage. 
The data centers distribute information to users for maintenance, auditing, future predictions etc. 
We solve the problem of privacy-protected aggregation at different levels like HAN, BAN, NAN, transmitting to  RTU and
then providing access control of the data stored at the data repository. 
The smart grid architecture is depicted in Figure \ref{fig:architecture}. 

The data aggregation network  on the left side of Figure \ref{fig:architecture} collects and aggregates data and sends to the RTU at the 
nearest substation. 
The right side of the Figure \ref{fig:architecture} shows the access control network, consisting of
RTUs, KDCs, data repository and users. 

Aggregation at each stage uses Paillier additive homomorphic encryption \cite{P99} which ensure that data can be aggregated only knowing the ciphertext, 
so that the plaintext can be hidden. 
This will protect the privacy of  individuals as well as a particular locality. 
Access control operates in the following way:
The RTU collects information from different units and sends to the data repository, encrypting the data under a set of  attributes.
Attributes of data can be the source of energy like solar, fossil-fuels, etc, or type of user like individual or corporate or plugged in vehicles, 
or the type of load like lower consumption equipment (as in lights, television) or high consumption equipment (dryers, heaters etc). 
The RTU can also add new attributes depending on the
time of collection (peak/offpeak time), type of user who can access (like engineers, environmentalist), location of the user (region/city) etc.
In this way the RTU builds an access policy for the data. 

The task of the KDC is to distribute keys to the RTU and users, such that the data is securely kept in the data repository and
retrieved only by authorized users. 
The KDCs can be energy management units who manage key distribution for attributes like source of energy, 
or power control units who look after key distribution for different types of equipments (like high-energy consumption or low energy consumption)
or administrative officers who distribute keys depending on the type of user that the RTU wishes to give access.
The KDCs also give keys to the users to enable them to decrypt messages, depending upon the attribute
they possess. 
For example if a environmentalist is interested in green energy (like solar/wind) then he/she is given keys corresponding to these attributes. 

Different users of data can access information stored in the databases, provided they have a valid access structure.
For example, a maintenance unit might want to collect information from residential and corporate users, which run on fossil fuel and
which have have consumed more than a given limit of electric power per day. 
Researchers on the other hand might be interested on predicting load due to charging and discharging of plugged in hybrid electric vehicles (PHEV)
during day time.
For each RTU, the key distribution centers distribute attributes and public and private keys.
The RTUs might have specific access policies. 
The RTUs encrypt the data with keys (depending on the access policy)  and sends to the storage units.

The data repository is responsible for both storing and processing information.
It can have several data storage centers. 
Processing can be done by one or several data processors, which can provide efficient search techniques or help organize the data in databases. 
We will not consider such aspects here. 
Users are also given attributes and secret keys. 
When users request data from the data repository, then they can decrypt those data that have matching attributes.
We apply a recent variant of attribute-based encryption, proposed by Lewko and Waters \cite{LW11}, modified according to the needs of smart grids.

Since a smart grid has a bidirectional flow of information, another feature can be added, in which users can send information to
selected RTUs.  
For example the maintenance units can ask certain RTUs  to reduce power consumption in certain  units on certain  weekends (for electrical maintenance)
or involve in a more complicated tasks.
Current RTUs are programmable and in future it would be possible to incorporate more features into them.

\subsection{Our contribution}
\label{subsec:contribution}
\begin{itemize}
\item We propose a new security architecture for smart grids, integrating privacy preserving aggregation and access control.
\item Aggregation of data at gateway smart meters of BAN, HAN, NAN is done using homomorphic encryption.
\item We propose an access control scheme which 
gives limited access to data users like audit teams, technical maintenance teams, engineers, environmentalists, research groups, policy makers, management groups, etc.
\item The scheme is collusion secure, in that no two users can collude and gain access to data they alone cannot avail.
\item Malicious and illegal users can be revoked .
\item We evaluate the performance and show it is feasible in the smart grids.
\item We provide a list of open problems not considered before and provide partial solution to these.
\end{itemize}

\subsection{Organization}
The paper is organized in the following way. 
We present related work on security and privacy issues of smart grids in Section \ref{sec:related}. 
This section also discusses Paillier's cryptosystem \cite{P99} and Lewko and Water's scheme \cite{LW11}. 
In Section \ref{sec:prelims}, we describe mathematical tools, network model and assumptions used in our work. 
We discuss data aggregation in details in Section \ref{sec:aggregation} and access control in Section \ref{sec:access-control}.  
In Section \ref{sec:analysis} we analyze the security and performance of our scheme and compare with existing ones. 
We present open problems in Section \ref{sec:open-problems} and conclude in Section \ref{sec:conclusion}. 

\section{Related work}
\label{sec:related}
In this section we first present related work on security in smart grids. 
We discuss previous work on homomorphic encryption and show why we chose Paillier's \cite{P99} homomorphic scheme for data aggregation in 
smart grids. 
Then we discuss several attribute based encryption techniques
to show why Lewko and Water's \cite{LW11} is most suited to access control in smart grids. 

\subsection{Security and privacy in smart grid}
\label{subsec:security-privacy}
As we noted in the introduction that security is an important aspect of smart grid, not only to protect from military threats
but also protect from misbehaviors of consumers and different service providers integrated into the grid. 
Security issues in smart grid mainly focus on authenticating customer, operators, and service providers. 
There are several components in smart grids like SCADA (Supervisory control and data acquisition), cellular and  mobile links, 
fiber optic cables etc. Security of each of these components is essential in securing the grid.
The cyber security requirements of smart grids have been outlined by the National Institute of Standards and Technology (NIST) \cite{NIST10}. 
To protect smart grids smart grid PKI infrastructure has been proposed. 
This infrastructure should provide certification to the various components and devices in the network.
Specific certification policies need to be issued \cite{MR10}. 
Device attestation (ensuring the validity of the device) is an important requirement, since an invalid device can collect and send wrong electricity readings
and can result in overloading and failure. 

It is also important to authenticate the message sent by devices and components in the network. 
Each device has an identity. 
Fouda \emph{et al.} \cite{FFKLS10} proposed a message authentication protocol for a Smart grid which has the following network 
structure: The home area network (HAN) consists of individual apartment units which collect and report to the building area network BAN, 
which further report to the neighborhood area network (NAN). There are gateway smart meters installed in each unit in a HAN, BAN and NAN
that collect information and send to the next level. 
The communication in HAN is done using IEEE 802.15.4 Zigbee radio communication \cite{ZIGBEE}. 
The authors propose authentication techniques using Diffie Hellman key agreement protocol, Sign-and-Mac (SIGMA) and Internet Key Exchange (IKEv2) \cite{IKE2}. 
The message authentication techniques has less communication overheads compared to the scheme proposed in \cite{KK10}. 
In \cite{CSSL11}, the authors proposed an authentication of metering messages which they claim to have less overheads. 

Privacy in smart grid has been extensively studies because of its importance. 
Though we want to know the amount of consumed data, we do not want to know the details, for example which user uses which appliance 
and at what time. 
This is to protect the privacy of the user. 
Studying the details of consumed data helps to deduce the behavioral pattern to a certain extent. 

In order to annonymize the metering data, Efthymiou and Kalogridis  \cite{EK10} proposed a third party key escrow policy and uses several pseudonymous 
IDs instead of unique identifiers. 
Rial and Danesiz \cite{RD10} proposed a privacy preserving protocol for smart meters using zero knowledge proof \cite{QQQQGGGGGGB89}, 
which ensures correct payment of fees with disclosing the details about consumption data. 
The protocol is implemented into smart meters and is generic enough to consider different billing settings like electronic traffic pricing, 
pay-as-you-drive car insurance etc.

However, imposing privacy policies can affect the utility. 
Very recently Rajagopalan \emph{et al.} \cite{RSMP11} quantify (using Gaussian model) how the utility is affected when privacy preservation is applied. 
They proposed that filtering out frequency components that are low in power can achieve an optimal utility-privacy solution. 
  
Access control has not been studied much, even though there is a big need for it. 
Bobba \emph{et al.} \cite{BKAA09} presented a  centralized access control scheme. 
As mentioned before, in the introduction,  these scheme has a drawback because centralized authority can be a single point of failure. 
It also requires that the KDC is online during data transfer. So the system is affected when the KDC is faulty or switched off for maintenance.  
ABE was discussed in connection to access control in smart grids in \cite{BKAA09}, but was not applied. 

We now present an overviews about homomorphic encryption and then ABE. 

\subsection{Overview of homomorphic encryption schemes}
\label{subsec:homomorphic}
The main idea for using homomorphic encryption is to carry out different operation on ciphertext and return results without knowing
the plaintext messages. It has been largely used in voting mechanisms where the individual votes should not be known but the decision is 
important. This is done in order to achieve privacy of the voter. 
Homomorphism has also been applied to data aggregation in ad hoc networks (for example \cite{CCMT09}).
Several encryption techniques exists which support different homomorphism, like multiplicative homomorphism 
(RSA \cite{RSA83}), additive homomorphism (Paillier \cite{P99}, Boneh-Goh-Nissim \cite{BGN05}) or recently proposed
fully homomorphic scheme \cite{G09} which can support complicated functions. 
During aggregation, we need to add the results as such we choose Paillier's cryptosystem, which supports additive homomorphism. 
Boneh-Goh-Nissim \cite{BGN05} is not a suitable choice because the set of messages in their system is very restrictive.

\subsection{Overview of attribute based encryption}
\label{subsec:abe}
ABE is a cryptographic protocol proposed by Sahai and Waters in 2005 \cite{SW05}.
The main idea is to distribute attributes to receivers and attributes to senders so that only receivers 
with matching  attributes structure can access the data.  
Data is encrypted using attribute based keys, which are distributed by a central key distribution center (KDC). 
It is to be noted that identity based encryption (IBE) proposed by Shamir \cite{S84} is a special form
of ABE, where senders have one unique attribute (i.e., its identity). 
The protocol proposed by Sahai and Waters was restricted that only threshold access structures ($t$-out-of-$n$) could be supported. 
This means that if the receiver has $t$ attributes (out of $n$) in common to the sender, then it can decrypt the message. 
Goyal \emph{et al.} \cite{GPSW06} proposed a new ABE which can handle any monotonic access structure. 
These schemes are known as key-policy based (KP-ABE) schemes. 

Another type of protocols are known as ciphertext-policy ABE (CP-ABE) \cite{BSW07} (proposed by Bethencourt). 
In these the ciphertext is encrypted using a set of attributes under a given access structure. 
If a receiver has a matching set of attributes then it can decrypt the information. 

All the above schemes relied on a central key and attribute distribution center, which is prone to failures. 
Chase \cite{C07} proposed a multi-authority (same as multi-KDC) protocol, where several KDCs generate and distribute
keys and attributes. There is also a central trusted authority who coordinates the multiple KDCs. 
To completely do away with central authority, Chase and Chow \cite{CC09} proposed a scheme where the authorities can coordinate amongst themselves, 
but do not require a central authority. 
The drawback of this protocol was that the access structure was specific and required each user to have at least one attribute from each KDC. 
Both these scheme were KP-ABE.

Recently Lewko and Waters \cite{LW11} proposed a multi-KDC CP-ABE, which does not have trusted authority and coordination between the KDCs.
It also allows any type of monotonic access structure. 
We use Lewko and Waters scheme to design an access control mechanism for smart grids.

\section{Background}
\label{sec:prelims}
In this section we present our network model and the assumptions we have used in the paper. 
Table \ref{table:notations} presents the notations used throughout the paper. 
We also describe mathematical background used in our proposed solution.

\begin{table}[ht]
\begin{center}
\caption{Notations}
\begin{tabular}{|c|l|}
\hline
Symbols & Meanings \\
\hline
$U_u$ & $u$-th User\\
$T_i$ & $i$-th RTU\\
$han_i$ & Gateway smart meter at $i$-th HAN \\
$ban_i$ & Gateway smart meter at $i$-th BAN\\
$nan$ & Gateway smart meter at NAN \\
$A_j$ & KDC $j$\\
$\mathcal{A}$ & Set of KDCs \\
$\mathcal{W}$ & Set of attributes \\
$w = |\mathcal{W}|$ & Number of attributes\\
$L_j$ & Set of attributes that KDC $A_j$ possesses\\
$l_j = |L_j|$ & Number of attributes that KDC $A_j$ possesses\\
$I[j,u]$ & Set of attributes that  $A_j$ gives to user $U_u$\\
$I_u$ & Set of attributes that user $U_u$ possesses\\
$PK[j]$ & Public key of KDC $A_j$ or RTU $T_j$\\
$SK[j]$ & Secret key of KDC $A_j$ or RTU $T_j$\\
$sk_{i,u}$ & Secret key given by $A_j$ corresponding to attribute $i$ \\
& given to user $U_u$ \\
$S$ & Boolean access structure\\
$R$ & Access matrix of dimension $n\times h$\\
$|G|$ & Order of group $G$\\
$M$ & Message\\
$P_j$ & Power consumption by gateway at $j$th HAN\\
$C$, $c$ & Ciphertext\\
$PKT[i]$ & Packet sent by smart meter gateway $i$\\
$H$ & Hash function, example SHA-1\\
\hline
\end{tabular}
\label{table:notations}
\end{center}
\end{table}

\subsection{Network model}
\label{subsec:network-model}

Our network model consists of two parts:
\begin{enumerate} 
\item First part is to collect data from consumers and aggregate them at different levels.
There are smart meters at each household which collect information about electrical usage by the consumer. 
This is the home area network (HAN).
The gateway smart meter processes data and sends to the smart meter at the BAN, which then aggregates and sends to the smart meter at NAN. 
The NAN gateway sends information to the substations. 
\item Second part is similar to currently deployed Control and Data Acquisition, and Energy Management
System (SCADA/EMS). It consists of remote terminal units which collect information from the NAN and other sources like PHEVs and
sends to the SCADA/EMS. 
In our model, SCADA/EMS consists of data repository which stores the data collected by the RTU. 
It also has data processors to process data. 
There are also key distribution centers who distribute keys to RTU and users.
This architecture consisting of RTUs, data repositories, KDCs and users is similar to \cite{KTKL10}, however they didn't address the problem of access control. 
Data aggregation was also not included. 
Users can be system engineers, maintenance offices, auditors, policy makers, researchers etc. 
\end{enumerate}
The architecture is presented in Figure \ref{fig:architecture}.

\subsection{Assumptions}
\label{subsec:assumptions}
We assume that each device has an identity (an IP address) and can authenticate itself before interacting with the network. 
We will not design an authentication protocol here, but rely on the authentication protocol \cite{FFKLS10}, which has been designed
specially for Smart grid communication. 

All the smart meters at the data aggregation centers of BAN, HAN and NAN are assumed to honest but curious. 
This means that they always send correct aggregated results, but would like to know the data that it receives from the previous 
smart meter aggregator. 
Hence, we assume that data aggregation smoothly, but there is a need to protect the consumer's privacy. 

We also assume that the data storage center is honest but curious. 
This means, it can attempt to read the contents of the ciphertext and the attributes that the ciphertext might be carrying. 
The RTUs are also honest but curious, so we hide the privacy of individual customers. 
However, we must remember that when the RTUs are sending messages, they can choose their access policies according to their discretions, 
depending upon the data they are sending. 

As mentioned earlier, attributes can be one or more of the following types (but not limited to)
\begin{enumerate}
\item Type of energy source: fossil fuel, solar, hydroelectricity, wind.
\item Type of consumer: Individual, corporate, PHEV.
\item Location of the consumer: City, region. 
\item Type of appliances: Need based. For example essential like light, heat etc. Lower priority: Dryer,  washing machine.
\item Load based: High electricity consumption equipments like dryer, oven etc, low electricity consumption equipments like lights, television etc.
\item Type of user: Electrical engineer, power engineer, environmentalists, policy makers etc.  
\end{enumerate}

These attributes do not reveal the identities of the users, because the RTU collect these information from the users and aggregate them. 
Hence, there is no risk of the maintenance offices, researchers, policy administrators to know the individual identity. 
Thus, individual's privacy is protected.

\subsection{Formats of access policies}
\label{subsec:access policy}
Access policies can be in either formats 1) Boolean functions of attributes or 2) Linear Secret Sharing Scheme (LSSS) matrix. 
Any access structure can be converted into a Boolean function \cite{LW11}. 
An example of a boolean function is $((a_1\wedge a_2 \wedge a_3)\vee(a_4\wedge a_5))\wedge(a_6\vee a_7))$, 
where $a_1, a_2, \ldots, a_7$ are attributes. 
Boolean functions can also be represented by access tree, with attributes at the leaves and $AND (\wedge)$ and $OR(\vee)$  
as the intermediate nodes and root. 
Our pseudo-code of an algorithm that converts a Boolean function (in the form of access tree) to a LSSS  matrix is given in the Appendix. 
The algorithm is described in \cite{LW11} as follows. Root has vector (1). Let $v[x]$ be parent's vector. 
If node $x$=AND, then the left child is $(v[x]|1)$, and the right child is $(0, \ldots, -1)$.
If $x$=OR, then both children also have unchanged vector $v[x]$.
Finally, pad with 0s in front, such that all vectors are of equal length.
The proof of validity of the algorithm is given in \cite{B96}. 
Fig. \ref{fig:tree} shows an access tree with initial vectors.  The rows of $R$ are the required vectors.

\subsection{Mathematical background}
\label{subsec:maths}
We will use bilinear pairings on elliptic curves. 
Let $G$ be a cyclic  group of prime order $q$ generated by
$g$. Let $G_T$ be a group of order $q$. 
We can define 
the map 
$e: G \times G \rightarrow G_T$. The map satisfies the following properties:
\begin{enumerate}
\item 
$e(aP, bQ) = e(P, Q)^{ab}$ for all $P, Q \in G$ and $a,b \in \mathbb{Z}_q$, $\mathbb{Z}_q = \{0,1,2,\ldots, q-1\}$.
\item Non-degenerate: $e(g, g)\not=1$.
 
\end{enumerate}

We use bilinear pairing on elliptic curves groups. We do not discuss the pairing functions which mainly use
Weil and Tate pairings \cite{PBC} and computed using Miller's algorithm \cite{M86}.
The choice of curve is an important consideration, because it determine the complexity of pairing operations. 
A survey on pairing friendly curves can be found in \cite{FCT10}. 
PCB library (Pairing Based Cryptography) \cite{PBC}  is a C library which is built above GNU GMP (GNU Math Precision) library
and contains functions to implement elliptic curves and pairing operations. 
The curves chosen are either MNT curves or supersingular curves. 

\subsection{Paillier homomorphic scheme}
\label{subsec:paillier}
In this section we discuss Paillier's \cite{P99} homomorphic scheme which we will use for secure data aggregation protocol.
We will first discuss the encryption protocol and show how it can be used to support homomorphism.
Let $i$ be the receiver for whom a message is intended. The protocol consists of three algorithms:

\begin{enumerate}
\item Key generation: This algorithm generates the public keys, and global parameters, given a security parameter.
Let $N = q_1q_2$, where $q_1$ and $q_2$ are primes.
Choose $g \in \mathbb{Z}_{N^2}^{\ast}$, such that $g$ has order a multiple of $N$ modulo $N^2$.
Let $\lambda(N) = lcm(q_1-1, q_2-1)$, where lcm represents least common multiple.
Then public key of $i$ is $PK[i] = (N,g)$ and secret key $SK[i] = (\lambda(N))$.

\item Encryption: Let $M \in \mathbb{Z}_N$ be a message.
Select a random number: $r \in Z_N^{\ast}$.
The ciphertext $c$ is given by
\begin{equation}
\label{eq:pill_en}
c = E(M) = g^{M}r^N \mod N^2
\end{equation}

\item Decryption:
To decrypt $c$, $M$ can be calculated as
\begin{equation}
\label{eq:pill_de}
M = D(c) = \frac{L(c^{\lambda(N)} \mod N^2)}{L(g^{\lambda(N)} \mod N^2)} \mod N,
\end{equation} 
where the $L-$function takes input from the set $\{u < N^2|u=1 \mod N\}$ and computes $L(u) = (u-1)/N$. 
\end{enumerate}

Additive homomorphism is demonstrated in the following way.
Suppose $c_1 = E(M_1)$ and $c_2 = E(M_2)$ are two ciphertexts, for $M_1,M_2 \in \mathbb{Z}_N$.
Then, $D(c_1.c_2 \mod N^2) = M_1+M_2 \mod N$.
Thus, the sum of the ciphertext can be obtained from the plaintext.

We note that $r^N$ is used only to make the homomorphic computation indeterministic, the same message can be encrypted into different ciphertexts, to
prevent dictionary attacks.
 
\subsection{Lewko-Waters ABE scheme}
\label{subsec:lewko-waters}
Lewko-Waters \cite{LW11} scheme consists of four steps: 1) System Initialization, 
2) Key and attribute distribution to users By KDCs 3) Encryption of message by sender 
4) Decryption by receiver. 

\subsubsection{System Initialization}
\label{subsec:system_init1}
Select a prime $q$, generator $g$ of $G$, 
groups $G$ and $G_T$ of order $q$,  
a map $e:G\times G \rightarrow G_T$, and
a hash function $H : \{0,1\}^* \rightarrow G$ which maps the identities of users to $G$.
The hash function used here is SHA-1 \cite{S06}. 
Each KDC $A_j \in \mathcal{A}$ has a set of attributes $L_j$. The attributes disjoint  ($L_i \bigcap L_j = \phi$ for $i\not=j$).
Each KDC  also chooses two random exponents $\alpha_i, y_i \in \mathbb{Z}_q$.
The secret key of KDC $A_j$ is 
\begin{equation}
SK[j] = \{\alpha_i, y_i, i \in L_j\}.
\end{equation}
The public key of KDC $A_j$ is published: 
\begin{equation}
PK[j] = \{e(g,g)^{\alpha_i}, g^{y_i}, i\in L_j\}.
\end{equation}

\subsubsection{Key generation and distribution by KDCs}
\label{subsec:key_gen1}
User $U_u$ receives a set of attributes $I[j,u]$ from KDC $A_j$, and corresponding secret key $sk_{i,u}$ for each $i \in I[j,u]$

\begin{equation}
\label{eq:static_secret_node}
sk_{i,u} = g^{\alpha_i}H(u)^{y_i},
\end{equation}
where $\alpha_i, y_i \in SK[j]$. 
Note that all keys are delivered to the user securely using the user's public key, such that only that user can decrypt it using its secret key. 

\subsubsection{Encryption by sender}
Sender decides about the access tree. 
LSSS matrix $R$ can be derived as described in \ref{subsec:access policy}.
Sender encrypts message $M$ as follows: 
\begin{enumerate}
\item Choose a random seed $s \in \mathbb{Z}_q$ and a random vector $v \in \mathbb{Z}_q^h$, with $s$ as its first entry; $h$ is the number of leaves in the access tree (equal to the number of rows in the corresponding matrix $R$).
\item Calculate $\lambda_x = R_x\cdot v$, where $R_x$ is a row of $R$
\item Choose a random vector $w \in \mathbb{Z}_q^h$ with 0 as the first entry.
\item Calculate $\omega_x = R_x\cdot w$
\item For each row $R_x$ of $R$, choose a random $\rho_x\in \mathbb{Z}_q$.
\item The following parameters are calculated:
\begin{equation}
\begin{array}{l}
C_0 = Me(g,g)^s\\
C_{1,x} = e(g,g)^{\lambda_x}e(g,g)^{\alpha_{\pi(x)}\rho_x}, \forall x\\
C_{2,x} = g^{\rho_x} \forall x\\
C_{3,x} = g^{y_{\pi(x)}\rho_x}g^{\omega_x} \forall x ,
\end{array}
\label{eq:ciphertext1}
\end{equation}
where $\pi(x)$ is mapping from $R_x$ to the attribute $i$ that is located at the corresponding leaf of the access tree.
\item The ciphertext $C$ is sent by the sender (it also includes the access tree via $R$ matrix):
\begin{equation}
\label{eq:ciphertext}
C = \langle R, \pi, C_0, \{C_{1,x}, C_{2,x}, C_{3,x}, \forall x\}\rangle
\end{equation}
\end{enumerate}

\subsubsection{Decryption by receiver}
\label{subsubsec:decrypt}
Receiver $U_u$ takes as input ciphertext $C$, secret keys $\{sk_{i,u}\}$, group $G$, and outputs message $M$. 
It obtains the access matrix $R$ and mapping $\pi$ from $C$. 
It then executes the following steps:
\begin{enumerate}
\item $U_u$ calculates the set of attributes $\{\pi(x): x \in X\}\bigcap I_u$ that are common to itself and  
the access matrix. $X$ is the set of rows of $R$. 
\item For each of these attributes, it checks if there is a subset $X'$ of rows of $R$, such that
the vector $(1,0\ldots,0)$ is their linear combination.
If not, decryption is impossible. If yes, it
calculates constants $k_x \in \mathbb{Z}_q$, 
such that $\sum_{x \in X'} k_xR_x = (1,0,\ldots,0)$. 
$K$ is a vector consisting of $k_x$, $x\in X'$. 
\item Decryption proceeds as follows:
\begin{enumerate}
\item  
For each $x\in X'$,
$dec(x) = \frac{C_{1,x}e(H(u),C_{3,x})}{e(sk_{\pi(x),u},C_{2,x})}$
\item $U_u$ computes $M = C_0/\Pi_{x\in X'} dec(x)$.
\end{enumerate}
\end{enumerate}

\section{Secure aggregation by smart meters}
\label{sec:aggregation}
In this section we discuss how aggregation takes place at the gateway smart meters $han$, $ban$ and $nan$ 
before it reaches the substation. 
We assume that the following architecture exists:
The household meters collect samples the readings from different equipments and sends to the gateway smart meter at the HAN. 
The gateway smart meters $han$ send their aggregated results and send to the $ban$. 
The gateway  smart meter $ban$, aggregates all the readings from the gateway meters at HAN meters and sends to the NAN. 
The gateway HAN smart meter aggregates all the readings from the gateway BANs  and sends to the nearest substation.  
This is depicted in Figure \ref{fig:architecture} (left side). 

An RTU $T_i$ is securely given  $PK[i] = (N,g)$ (as in key generation step in Section \ref{subsec:paillier}) 
and also the secret key $SK[i]=\lambda(N)$.  
Each smart meter in the network knows the public key $PK[i] = (N,g)$ of its nearest RTU substation $T_i$. 
Each gateway smart meter $han_j$  sends a data packet which consists of two fields: the attributes field $f$ and the power consumption field $P_j$. 
The power consumption field is encrypted with the public key of the substation. 
A packet looks like 
\begin{equation}
PKT[han_j] = f||c_j = f||E(P_j), 
\end{equation}
where $E(P_j) = g^{P_j}r_j^N \mod N^2$ ($r_j\in \mathbb{Z}_N^{\ast}$ is chosen randomly by the smart meter). 

This packet is then send to the gateway BAN, $ban_l$ which aggregates all the results. 
Here it checks for the attributes field. 
For packets which have the same set of attributes, it processes the aggregated power consumption. 
The aggregated result is given by
$c_{ban_l} = \Pi_{j\in HAN}c_j$.  
The new packet looks like $PKT[ban_l] = f||c_{ban_l}$.

The packets collected by the gateway BANs are then send to the NAN. 
It performs a similar operation and aggregates information from packets having same set of attributes. 
The aggregated result is 
$c_{nan} = \Pi_{ban_l\in BAN}c_{ban_l}$. 
The packet $PKT[nan] = f||c_{nan}$ is then sent to the nearest substation. 

The RTU $T_i$ at the substation reads the content of the packet. 
It then decrypts the aggregated result because it has the secret key $SK[i]$. 

We note that \\
$c_{nan} = \Pi_{ban_l\in BAN}(\Pi_{j\in HAN}c_j)$ \\
\\
	$= \Pi_{ban_l\in BAN}(g^{\sum_{j\in HAN}P_j})(\Pi_{j\in HAN}r_j)^N \mod N^2$\\
\\
	$= g^{\sum_{j}P_j}(\Pi_{j}r_j)^N \mod N^2$

Using the value of $\lambda(N)$, the aggregated message can be decrypted by the RTU (as given in Section \ref{subsec:paillier}).

We next consider a very small example to show how this works in practice.  

\subsection{Example}
We show only the data having same set of attributes. 
The aggregation network is shown in the Figure \ref{fig:nan}. 

The HANs collect data from different devices and the encrypted data $c_1, c_2, \ldots, c_5$ to the respective BANs. 
Here $c_i = g^{P_i}r_i^N \mod N^2$, for $i = \{1,2\ldots,5\}$.
The BAN gateways aggregate the results. 
$ban_1$ calculates 
\begin{center}
$c_{ban_1} = c_1c_2 =  g^{P_1+P_2}(r_1r_2)^N \mod N^2$,
\end{center} 
 while
$ban_2$ calculates 
\begin{center}
$c_{ban_2} = c_3c_4c_5 =  g^{P_3+P_4+P_5}(r_3r_4r_5)^N \mod N^2$.
\end{center} 
The BAN gateways then send to the NAN, which aggregates the result as\\

\begin{center}
$c_{nan} = c_{ban_1}c_{ban_2} = c_1c_2c_3c_4c_5$\\
\vspace*{.1cm}
$= g^{P_1+P_2+P_3+P_4+P_5}(r_1r_2r_3r_4r_5)^N \mod N^2$
\end{center}

When RTU receives ciphertext $c_{nan}$, then decrypts it suing its secret key $\lambda(N)$ as\\
\vspace*{.1cm}
$D(c_{nan}) = \frac{L(c_{nan}^{\lambda(N)}) \mod N^2)}{L(g^{\lambda(N)}) \mod N^2)}\mod N$\\
\\
$~~~~~= \frac{L(g^{(P_1+P_2+P_3+P_4+P_5)\lambda(N)}) \mod N^2)}{L(g^{\lambda(N)}) \mod N^2)}\mod N$\\
\\
$~~~=P_1+P_2+P_3+P_4+P_5$.\\
This is because $(r_1r_2r_3r_4r_5)^{N\lambda(N)} = 1 \mod N^2$.

\begin{figure}[htb]
\begin{centering}
\includegraphics[width=3in]{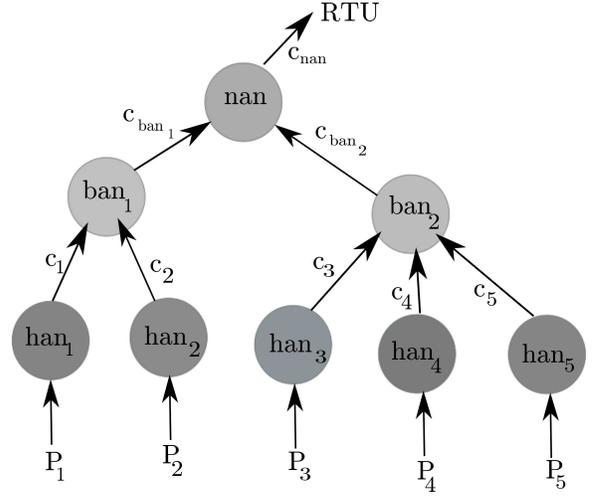}
\caption{
Example showing data aggregation
}
\label{fig:nan}
\end{centering}
\end{figure}

\section{Access control scheme}
\label{sec:access-control}
We will first provide a sketch of the scheme and then discuss it in details.  

The parameters are chosen and distributed to the KDCs when they are installed. 
The attributes and key generation has been presented in \ref{subsec:lewko-waters}.

Encryption proceeds in two steps. The Boolean access tree is first converted to LSSS matrix. 
In the second step the message is encrypted and sent to the data storage center along with the LSSS matrix. 
A secure channel like ssh can be used for the transmission. 

Suppose an RTU $T_i$ wants to store a record $M$. 
$T_i$ defines the access structure $S$, 
which helps it to decide the authorized set of users, who can access the record $M$. 
It then creates a $m \times h$ matrix $R$ ($m$ is the number of attributes in the access structure) and defines a mapping function 
$\pi$ of its rows with the attributes (using Algorithm in Section \ref{subsec:access policy}). 
 $\pi$ is a permutation, such that $\pi: \{1,2,\ldots,m\} \rightarrow  \mathcal{W}$. 
The encryption algorithm takes as input the data  $M$ that needs to be encrypted, the group $G$, the LSSS matrix $R$, the permutation function $\pi$, 
which maps the attributes in the LSSS to the actual set of attributes. 
For each message $M$, the ciphertext $C$ is calculated as per the Equations (\ref{eq:ciphertext1}) and (\ref{eq:ciphertext}). 
Ciphertext $C$ is then stored in the data repository. 

When a user $U_u$ requests a ciphertext from the repository, 
the requested ciphertext $C$ is transferred using ssh protocol. 
The decryption algorithm proceeds as in Section \ref{subsubsec:decrypt}, and returns plaintext message $M$, if the user has valid set of attributes.

\subsection{An Example}
\label{subsec:example}
Suppose an RTU sends a data record to the data repository. 
This data can be the amount of electricity consumed over a certain period of time by high-consumption equipments which are run by fossil fuels. 
The RTU can give access to either researchers and policy makers or give selective access to environmentalist working on fossil-fuels
or power engineers who are supervising the usage of high-consumption equipments. 
There can be three types of KDC:
1)Type of users: $D_1$ (Researchers) , $D_2$ (policy makers), $D_3$ (Power engineers), $D_4$ (Environmentalists), etc,
2)Type of appliance: $E_1$ (High consumption), $E_2$ (Low consumption) etc,  
3)Source of power: $S_1$ (fossil-fuels), $S_2$ (solar), etc. 

Then the access tree is given in Figure \ref{fig:tree}. 
\begin{figure}[htb]
\begin{centering}
\includegraphics[width=3.0in]{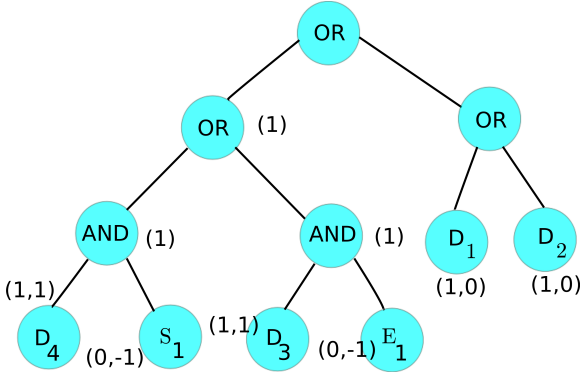}
\caption{
Access tree structure
}
\label{fig:tree}
\end{centering}
\end{figure}

The access matrix $R$  can be constructed using Algorithm in Appendix.  
Thus, 
\begin{center}
$R = 
\begin{pmatrix}
1 & 1 \cr
0 & -1 \cr
1 & 1  \cr
0 & -1 \cr
1 & 0  \cr 
1 & 0  \cr
\end{pmatrix} .
$
\end{center}
An environmentalist working on fossil fuels will be able to access this data, as also a power engineer monitoring high-consumption equipments. 
However an electrical engineer working on solar cells will not be able to read it.

Let there be three KDCs $A_1, A_2$ and $A_3$. 
The set of attributes of $A_1, A_2$ and $A_3$ are
$L_1 = \{D_1, D_2, D_3, \ldots\}$ and
$L_2 = \{E_1, E_2, \ldots\}$ and 
$L_3 = \{S_1, S_2, \ldots\}$.
The RTU's access tree is given by  Fig. \ref{fig:tree}. 
Let $\pi$ be denoted as
\begin{center}
\begin{tabular}{c|c|c|c|c|c|c}
$x$ & 1 &2 &3 &4 &5 &6\\
\hline
$\pi(x)$ & $D_4$ & $E_1$ & $D_3$ & $S_1$ & $D_1$ &$D_2$
\end{tabular} .
\end{center}

Suppose an user (user $u = 3$) is an environmentalist studying fossil-fuels and solar energy, then 
he/she  is given the attributes $D_4$, $E_1$ and $E_2$. 
Thus, $I[1,3] = \{D_4\}$ and 
$I[2,3] = \{S_1,S_2\}$. 
Next the user is given secret keys $sk_{4,1}$ from $A_1$ and $sk_{1,3}$ and $sk_{2,3}$ from $A_3$.

During encryption, the RTU sends the information  $C = \langle R, \pi, C_0, \{C_{1,x}, C_{2,x}, C_{3,x}\}_{x \in \{1, 2, 3, 4, 5,6\}}\rangle$ to the
data repository.  
$C_0 = Me(g,g)^s$, where $s$ is chosen at random from $\mathbb{Z}_q$. 

When user 3 wants to access the above information $C$. $C$ is transferred securely, using ssh (an inbuilt secure shell standard protocol). 
The user first finds out the attributes that are present from $\pi$. 
He/she also finds that it has the attributes $D_4, S_1$ in common to the attribute in data. 
From the matrix $R$ it then finds that there are two rows corresponding to $D_4$ and $S_1$, such that
$(1, -1)+ (0,1) = (1,0)$ (linear combination of rows 1 and 2 of $R$ gives $(1,0)$). 

The user can thus calculate $e(g,g)^s$ according to Step 4 of the decryption mechanism. 
Once $e(g,g)^s$ is calculated, $M$ can be obtained. 
The data repository does not have the secret keys, and is unable to decrypt the message.

\subsection{Revocation of users}
\label{subsec:revoke}
Users can be revoked, either because they are faulty or have been tampered with. 
Once revoked, these should not be able to decrypt messages, even if they have valid attributes. 
We present a revocation mechanism to achieve this. 

For each revoked user $U_u$, $I_u$  is noted. 
Once the attributes $I_i$  are identified,
all data that possess the  attributes are collected.
For each such information record, the following steps are then carried out:
\begin{enumerate}
\item A new value of $s$, $s_{new} \in \mathbb{Z}_q$ is selected
\item The first entry of vector $v_{new}$ is changed to new $s_{new}$
\item $\lambda_x = R_x v_{new}$ is calculated, for each $x \in I_i$
\item $C_{1,x}$ is recalculated for $x \in I_i$
\item New value of $C_{1,x}$ is securely transmitted to the storage center
\item New $C_0 = Me(g,g)^{s_{new}}$ is calculated and stored in the storage center
\item New value of $C_{1,x}$ is not stored with the data, but is transmitted to users, who wish to decrypt the data.
\end{enumerate}

We note here that the new value of $C_{1,x}$ is not stored in the data centers but transmitted to the non-revoked users who have attribute $x$. 
This prevents a revoked user to decrypt the new value of $C_0$ and get back the message.

\section{Analysis and performance}
\label{sec:analysis}

\subsection{Security of aggregation mechanism}
\label{subsec:security-aggregation}
We will first show that the aggregation scheme gives correct results when the intermediate smart meter (HAN, BAN, NAN gateway) is honest.
We will then prove that the privacy of not only individual customers but also that of intermediate smart meters in BAN and NAN is preserved.

\begin{theorem}
The aggregation scheme presented in Section \ref{sec:aggregation} gives correct results when the intermediate smart meter (HAN, BAN, NAN gateway) is honest.
\end{theorem}

\begin{proof}
We first note that the decryption step given in Equation \ref{eq:pill_de} is correct. 
$c^{\lambda(N)} \mod N^2$ and $g^{\lambda(N)} \mod N^2$ both equal 1, when raised to the power of $N$. This is because $g$ has an order which is a 
multiple of $N$. 
Thus,  $c^{\lambda(N)} \mod N^2$ and $g^{\lambda(N)} \mod N^2$ are both $N$-th roots of unity. 
Such  roots are of the form $(1+N)^{\beta} = (1 + \beta N) \mod N^2$. 
Hence, the $L$-function can be computed as $L((g^M)^{\lambda(N)} \mod N^2) = M L(g^{\lambda(N)} \mod N^2) \mod N$. 
(details of proof appear in \cite{P99}).
From this, the value of $M$ can be obtained. 

For our aggregation scheme, 
\begin{center}
$c_{nan}  = g^{\sum_{j}P_j}(\Pi_{j}r_j)^N \mod N^2$.   
\end{center}
We note that $((\Pi_{j}r_j)^N)^{\lambda(N)} = 1 \mod N^2$. 
Thus, 
\begin{center}
$D(c_{nan}) = \frac{L((g^{\sum_{j}P_j})^{\lambda(N)} \mod N^2}{g^{\lambda(N)}\mod N^2}\mod N$  \\
\vspace*{.1cm}
$~~~~~~~~~~~~= \sum_{j}P_j$, (by similar argument as above). 
\end{center}
\end{proof}

\begin{theorem}
Data aggregation scheme proposed in Section \ref{sec:aggregation} protects the privacy of customers and all nodes in BAN and HAN. 
\end{theorem}

\begin{proof}
Pailler's cryptosystem is intractable under Decisional Composite
Residuosity Assumption (DCRA) \cite{P99}. 
A customer sends encrypted data of its power consumption. 
The data is encrypted using public key of the nearest substation. 
As such no user or outsider can decrypt the data unless it knows $\lambda(N)$ which is difficult to solve.  

Next, we note that even the RTU at the substation cannot know the individual ciphertexts. 
This is because it receives encrypted aggregated results from which individual ciphertexts cannot be obtained. 
The use of the factor $r^N$ while encrypting message ($r$ chosen randomly for each message) helps to transmit
the same message as two different ciphertexts and thus prevents dictionary attacks. 

Thus, no user/substation can decrypt data that an individual customer sends, thus protecting privacy. 
\end{proof}

\subsection{Security of our access control scheme}
\label{subsec:security-abe}
We will show that only authorized users (possessing valid set of attributes) can decrypt the data stored in data repositories. 
The data center cannot change the content of the data stored in the data bases. 
The data center cannot collude with an user or RTU and decrypt any information it is not supposed to decrypt.
No two users can share their attributes and secret keys and decrypt any information they are not supposed to decrypt alone. 

\begin{theorem}
The proposed access control scheme is secure, collusion resistant, allows access of data only to authorized users and protects the privacy of 
individual consumers.
\end{theorem}

\begin{proof}
We will first show that a user can decrypt data  if and only if it has a matching set of attributes. 
This follows from the fact that access structure $S$ (and hence matrix $R$) is constructed 
if and only if
there exists a set of rows $X'$ in $R$, and linear linear constants $k_x\in \mathbb{Z}_q$, such that $\sum_{x\in X'}k_xR_x = (1,0,\ldots,0)$.  
A proof of this appear in \cite[Chapter 4]{B96}.
For an invalid user, there does not exists attributes $x$, such that  
$\sum_{x\in X'}k_xR_x = (1,0,\ldots,0)$. 
Thus, $e(g,g)^s$ cannot be calculated. 
Hence, our scheme allows access of data only to authorized users. 

We next show that two or more users cannot collude and gain access to data that they are not individually supposed to access. 
Suppose that there exist attributes $\pi(x)$ from the colludes, such that $\sum_{x\in X}k_xR_x = (1,0,\ldots,0)$. 
However, $e(H(u),g)^{\omega_x}$ needs to be calculated 
in Section \ref{subsubsec:decrypt}.
Since different RTUs different values of $e(H(u),g)$, 
even if they combine their attributes, 
they cannot decrypt the message. 
Thus, our access control scheme is collusion secure.

We next observe that no outsider or even the data center administrator can decrypt any information stored in the databases. 
This is because an outsider or a data center administrator  does not posses the secret keys $sk_{i,u}$ (by Eq.(\ref{eq:static_secret_node})). 
Even if they collude with other users, they cannot decrypt data which the users cannot themselves decrypt, because of the
above reason (same as collusion of users).  
The KDCs are work autonomously and are not a part of the data center. 
Thus, no outsider can decode data stored in the repositories, without compromising the relevant KDCs.  
This makes our scheme secure. 

The RTUs receive aggregated results from the HAN, BAN and NAN. 
The consumers send encrypted data and it is never decrypted at any stage. 
This protects the privacy of consumer's data.

\end{proof}

\begin{table*}
\caption{Comparison of our scheme  with Bobba \emph{et al.}\cite{BKAA09}}
\begin{center}
{\small
\begin{tabular}{|c|c|c|c|c|}
\hline
Schemes & Robustness & Access policy & Revocation  & Online/offline  \\
 &  &  & possible or not & KDC  \\
\hline
Bobba \emph{et al.} \cite{BKAA09} & Not robust & Any  & Yes & Has to remain\\
& Centralized & boolean function &  & online   \\
 & administration &  &  &   \\
\hline
Our scheme & Robust & Any monotonic  & Yes & Need not be online \\
& distributed KDC &  boolean function &  &  \\
\hline
\end{tabular}
}
\end{center}
\label{table:comp}
\end{table*}

\subsection{Performance issues}
\label{subsec:performance}

We will first calculate the cost of aggregation. 
Encryption involves modular exponentiation of element $g$, which can be done using square-and-multiply technique in $O(\log N)$ time. 
Decryption involves calculating $L(u)$, which needs only one multiplication. 
Decryptions can be hastened using the technique already given in \cite{P99}. 
At each smart meter gateway $d$ values have to be multiplied (where $d$ is the indegree of that smart meter). 
So the costs are reasonable.

We will calculate the computation and communication overhead of access control scheme with and without user, RTU revocation.
In the first step of encryption, the access tree needs to be converted to an access matrix. 
Time taken to compute $R$ from $S$ is $O(m)$, where $m$ is the number of attributes in the access structure.
To check if there exists a set of rows in $R$ (such that step (2) of decryption holds), is equivalent to solving the equation
$KR = (1, 0, \ldots,0)$, for non-zero row vector $K$.
This takes $O(mh)$.
Since the list of attributes might not be too large, such overhead is very little. 

The most expensive operation during encryption or decryption is pairing. 
During encryption,  each user $U_u$ performs only one pairing operation (to calculate $e(g,g)$).
For each row $x$ corresponding to attribute, it also performs two scalar multiplications to calculate $C_{1,x}$,
one scalar multiplication to calculate
$C_{2,x}$ and one to calculate $C_{3,x}$.
Thus, there are a total of $4m$ scalar multiplications.
During decryption, there are two pairing operations, one for $e(H(u),C_{3,x})$ and the other for $e(sk_{i,u},C_{2,x})$,
for each $x$.
The number of pairing operations is thus $2m$ to calculate $e(H(u),C_{3,x})$.
There are also at most $m$ scalar multiplications to calculate $(e(g,g)^{\lambda_x}e(H(u),g)^{\omega_x})^{\sigma_x}$.
Therefore, the computation time is $(2m + 1)T_p + 5mT_m$, where $T_p$ and $T_m$ are  the time taken to perform pairing and scalar multiplication.

Using PCB library (Pairing Based Cryptography) \cite{PBC} with an MNT curve of embedding degree $k=6$ and $q=160$ bit curve, 
$T_{mul} = 0.6ms$ and $T_p = 4.5 ms$. 
For an access policy consisting of 10 attributes, decryption time at each user is 124.5 ms. 
The decryption time increases linearly with the number of attributes in the access policy.

Information to be sent from RTU to data repository, and from the storage centers  to user require
$m\log |G_T| + 2m\log|G| + m^2 +  |Data|$ bits, where $|Data|$ is the size of the data.
$m^2$ bits are needed to transfer the matrix $R$, and $m(|G_T| + 2|G|)+ |G_T|$ to transfer $C_0$, $C_{1,x}$, $C_{2,x}$ and $C_{3,x}$ and
$\log w$, to send $\pi$.
Thus, the communication overhead is $m^2 + m(|G_T| + 2|G|)+ |G_T| + \log w + |Data|$.

When revocation is required, $C_0$ needs to be recalculated.
$e(g,g)$ is previously calculated. So, only one scalar multiplication is needed.
If the user revoked is $U_u$, then for each $x $, $C_{1,x}$ has to be recomputed.
$e(g,g)$ is already computed. Thus, only two scalar multiplication needs to be done, for each $x$.
So a total of $2m'+1$ scalar multiplications are done by the KDCs, where $m'$ is the number of attributes  belonging to all revoked users.
Users need not compute any scalar multiplication or pairing operations.
Additional communication overhead is $O((m'+1)|G_T|)$.

\subsection{Comparison with other schemes}
\label{subsec:comparison}
In this section we compare our access control scheme with that of Bobba \emph{et al.} \cite{BKAA09}. 
We show (in Table \ref{table:comp}) that our scheme is more robust than theirs, because ours is a decentralized scheme.
The biggest drawback of Bobba \emph{et al.} \cite{BKAA09} is that the centralized KDC has to be online
all the time to allow access of data. This is a huge restriction, 
because the system will completely shut off in case of fault or even maintenance.

\section{Open problems in smart grid security}
\label{sec:open-problems}
The data center stores huge amounts of data and thus maintenance of these databases can be a huge concern. 
One recent proposal is to integrate smart grids with clouds. 
In this context, cloud can provide infrastructure to store this huge amount of data. 
There has been quite a lot of research in information secure information retrieval using searchable encryption \cite{SWP00,CGKO06},
where searching is done checking the indices of the encrypted keywords. Result is returned 
without knowing the keyword or the retrieved record. 

Cables can be incorporated into the smart grid system to enable users to get efficient access of content. 
Content distributors can either provide satellite radio subscriptions (for channels for a fixed duration like a month of a year), or
provide impulse pay-per-view facility (viewers pay as and when they view a program like in hotels), 
prepaid pay-per-view (viewers pay in advance as in for a hockey match or a concert), 
pay-per-channel (viewers pay subscribe for a channel). 
There are several security and privacy issues that need to be addressed here. 
Efficient access control is very important because of the large number of viewers and attributes involved. 
Previous work on content access control has been done by Pirretti \emph{et al.} \cite{PTMW10}.
The question is how to efficiently integrate them into the grid. 
The other issue is of privacy, such that a viewers identity if not revealed at any time. 
This might give valuable information about the behavior of the individual. 

\section{Conclusion}
\label{sec:conclusion}
In this paper we have presented an secure architecture in smart grids which integrates aggregation and access control. 
Homomorphic encryption is used to preserve customer privacy, while ABE is used for achieving access control. 
ABE has not been used in access control in smart grids, though it has been mentioned as a possibility in \cite{BKAA09}. 
The access control architecture is decentralized, which makes it more attractive and practical than \cite{BKAA09}. 
We have also addressed a few open problems that can be worked on in  future. 


\end{document}